\documentclass[10pt,conference]{IEEEtran}
\IEEEoverridecommandlockouts

\usepackage{graphicx}
\usepackage{amssymb}
\usepackage{amsmath}
\usepackage{cite}
\usepackage{subfigure}
\usepackage{mathrsfs}
\usepackage[displaymath,mathlines]{lineno}
\usepackage{color}
\usepackage{tabulary}
\usepackage{multirow}
\usepackage{algpseudocode}
\usepackage{algorithm,algpseudocode}
\usepackage{algorithmicx}
\usepackage{pbox}
\usepackage{multicol}
\usepackage{lipsum}
\usepackage{amsthm}
\usepackage{relsize}
\usepackage{lipsum}
\usepackage{epstopdf}
\usepackage{mathtools}% http://ctan.org/pkg/mathtools
\usepackage{calc}% http://ctan.org/pkg/calc
\usepackage{tabularx}

\makeatletter
\newcommand{\doublewidetilde}[1]{{%
		\mathpalette\double@widetilde{#1}}}
\newcommand{\double@widetilde}[2]{%
	\sbox\z@{$\m@th#1\widetilde{#2}$}%
	\ht\z@=.5\ht\z@
	\widetilde{\box\z@}}
\makeatother

\newtheorem{lemma}{Lemma}

\begin{document}	
	\title{Learning to Perform Downlink  Channel Estimation in Massive MIMO Systems}
	\author{\IEEEauthorblockN{Amin Ghazanfari\IEEEauthorrefmark{1}, Trinh Van Chien\IEEEauthorrefmark{2}, Emil Bj\"ornson\IEEEauthorrefmark{1}\IEEEauthorrefmark{3}, Erik G. Larsson\IEEEauthorrefmark{1}}
		\IEEEauthorblockA{\IEEEauthorrefmark{1}Department of Electrical Engineering (ISY), Link\"oping University, Sweden }
		\IEEEauthorblockA{\IEEEauthorrefmark{3}Department of Computer Science, KTH Royal Institute of Technology, Sweden}
		\IEEEauthorblockA{\IEEEauthorrefmark{2}Interdisciplinary Centre for Security, Reliability and Trust (SnT), University of Luxembourg, Luxembourg }
		\vspace*{-.75cm}
		\thanks{This paper was supported by ELLIIT and the Grant 2019-05068 from the Swedish Research Council.}
		}

\maketitle
	\begin{abstract}
We study downlink (DL) channel estimation in a multi-cell Massive multiple-input multiple-output (MIMO) system operating in a time-division duplex. The users must know their effective channel gains to decode their received DL data signals. A common approach is to use the mean value as the estimate, motivated by channel hardening, but this is associated with a substantial performance loss in non-isotropic scattering environments. We propose two novel estimation methods. The first method is model-aided and utilizes asymptotic arguments to identify a connection between the effective channel gain and the average received power during a coherence block. The second  one is a deep-learning-based approach that uses a neural network to identify a mapping between the available information and the effective channel gain.  We compare the proposed methods against other benchmarks in terms of normalized mean-squared error and spectral efficiency (SE). The proposed methods provide substantial improvements, with the learning-based solution being the best of the considered estimators.
	\end{abstract}

	\IEEEpeerreviewmaketitle

	%%%%%%%%%%%%%%%%%%%%%%%%%%
	%%%%%% Introduction %%%%%%
	%%%%%%%%%%%%%%%%%%%%%%%%%%
	\section{Introduction}
Massive multiple-input multiple-output (MIMO) is one of the backbone technologies for 5G-and-beyond networks \cite{larsson2014massive,Parkvall2017a}. In Massive MIMO, each base station (BS) is equipped with many active antennas to facilitate adaptive beamforming towards individual users and spatial multiplexing of many users \cite{marzetta2010noncooperative}. In this way, the technology can improve the spectral efficiency (SE) for individual users and, particularly, increase the sum SE in highly loaded networks by orders of magnitude compared with conventional cellular technology with passive antennas \cite{bjornson2017massive}. Having accurate channel state information (CSI) is essential in Massive MIMO networks \cite{larsson2014massive}, so that the transmission and reception can be tuned to the user channels, to amplify desired signals and reject interference.
Time-division duplex (TDD) operation is preferable for CSI acquisition because the BSs can then acquire uplink CSI from the uplink pilot transmission and utilize the uplink-downlink channel reciprocity to transform it to downlink (DL) CSI \cite{bjornson2017massive}. In this way, the required pilot resources are proportional to the number of users but independent of the number of BS antennas.

To decode the DL signals coherently, the each user must estimate the effective DL channel gain, i.e., an inner product of the precoding vector and the channel vector. The user only needs to know this scalar, not the individual vectors, but its value varies due to channel fading.  In the prior Massive MIMO literature, the fading variations have been neglected, motivated by the channel hardening effect that dictates that the effective DL channel gain is close to its mean value when there are many antennas \cite{jose2011pilot,yang2013performance,redbook}. More precisely, the receivers use this mean value as their estimate of the effective channel gain.
However, the required number of antennas to observe channel hardening depends strongly on the propagation environment.
With spatially correlated fading, one might need hundreds of antennas to achieve the same hardening level as with ideal independent Rayleigh fading \cite[Fig.~2.7]{bjornson2017massive}.
Estimating the effective DL channel gain using the mean value will result in a significant SE loss when the hardening level is low \cite{ngo2017no}. 
	
Another estimation approach is that the BSs beamform some DL pilots along with data to assist the users in estimating the effective DL channel gains \cite{ngo2013massive}. %The required DL pilot resources are proportional to the number of users and will increase the CSI acquisition overhead. 
Even though this approach will improve the estimates of the effective DL channel gains, the SE might decrease due to the extra overhead \cite{ngo2017no}. 

A blind estimator of the effective DL channel gain was developed for single-cell Massive MIMO systems in \cite{ngo2017no}.
It uses only the DL data signals for estimation and the enabling factor is that the precoding is selected to make the effective DL channel gains (approximately) positive and real-valued so that only the amplitude must be estimated. However, the method developed in  \cite{ngo2017no} relies on asymptotic arguments that are hardly satisfied in the operational regime of practical systems.
Nevertheless, the method was shown to perform better than the use of DL pilots since the blind channel estimation method does not need any extra pilots. The blind estimation method was generalized in  \cite{pasangi2020blind} to a multi-cell Massive MIMO network with uncorrelated Rayleigh fading channel and maximum ratio (MR) precoding at the BS. 

In this paper, we propose two new blind estimators of the DL channel gain for multi-cell Massive MIMO systems with correlated Rayleigh fading, MR precoding, and generic pilot assignment among the users. The first one is a multi-cell extension of the model-aided method from \cite{ngo2017no}. The second one is deep-learning-based and motivated by the fact that blind estimation builds on identifying a mapping between received data signals and the variable that is to be estimated. 
This coincides with the deep learning methodology of learning mappings between input signals and desired variables based on training data \cite{o2017introduction}. Neural networks have previously been used for developing physical-layer algorithms for interference management \cite{Sun2018a}, power control \cite{sanguinetti2018deep}, channel estimation \cite{demir2019channel}, among others. However, the application considered in this paper is novel.
Deep learning methods are particularly suitable for solving problems where the existing models are inaccurate or intractable for analytic development of algorithms, as is the case for the problem considered in this paper.

	\section{System Model}\label{sec:system-model}
	We consider a multi-cell Massive MIMO system with $L$ cells. Each cell has a BS equipped with $M$ antennas and serves $K$ single-antenna users. We use the conventional block
fading to model the randomness of the wireless channels over time
and frequency \cite[Sec.~2]{redbook}. The size of a coherence interval  is denoted $\tau_c$. The channel between BS~$l$ and user~$k'$ in cell~$l'$, which follows a correlated Rayleigh fading model is 
\begin{equation}
\mathbf{g}_{l'k'}^l \sim \mathcal{CN}\left({\mathbf{0}}, \mathbf{R}_{l'k'}^l \right),
\end{equation}
where ${\mathbf{R}}^{l}_{l'k'} \in \mathbb{C}^{M\times M}$ is the positive semi-definite spatial correlation matrix of the channel and $\beta^{l}_{l'k'}  = \mathrm{tr}\left(\mathbf{R}^{l}_{l'k'}\right)/M$, where $\beta^{l}_{l'k'} \geq 0$ is the corresponding average large-scale fading coefficient among the $M$ antennas.

We focus on the DL data transmission in a network operating with a TDD protocol. A new independent channel realization appears in every coherence interval.
	To enable spatial multiplexing in the DL, the BS must estimate the channels of the intra-cell users in every coherence interval and construct precoding vectors based on them. 
	We assume that this is done via uplink channel estimation at each BS once per coherence block. Each user transmits a pilot sequence from a predefined set of orthogonal pilots. We assume the cells share pilots using a pilot reuse factor of $f \geq 1$, which means the users within each cell have mutually orthogonal pilots, and the same pilot sequences are reused in a fraction $1/f$ of the $L$ cells in the network. To achieve this, we assume there is a set of $\tau_p = fK $ mutually orthogonal pilot sequences, each of length $\tau_p$. The channel estimation phase follows the standard minimum mean square error (MMSE) estimation approach in the literature, and the detailed derivation can be found in \cite[Theorem~3.1]{bjornson2017massive}. The MMSE estimate of $\mathbf{g}_{lk}^l$ is
	\begin{equation}
	\hat{\mathbf{g}}_{lk}^l = \sqrt{\hat{p}_{lk}} \mathbf{R}_{lk}^l \pmb{\Psi}_{lk}^{-1} \tilde{\mathbf{y}}_{lk},
	\end{equation}
	{\color{black}where $\tilde{\mathbf{y}}_{lk}$ is received pilot signal at BS $l$ from user $k$} and $\pmb{\Psi}_{lk} = \sum_{l'\in \mathcal{P}_l} \tau_p  \hat{p}_{l'k} \mathbf{R}_{l'k}^l + \sigma_{\mathrm{UL}}^2 \mathbf{I}_M $ and $\hat{\mathbf{g}}_{lk}^l$ is distributed as $\hat{\mathbf{g}}_{lk}^l \sim \mathcal{CN} \left(\mathbf{0}, \tau_p \hat{p}_{lk} \mathbf{R}_{lk}^l \pmb{\Psi}_{lk}^{-1} \mathbf{R}_{lk}^l \right)$.  Note that $\hat{p}_{l'k'}$ denotes the pilot power used by user $k'$ in cell $l'$ and $\mathcal{P}_{l}$ denotes the set of cells sharing the same subset of $K$ orthogonal pilot sequences as cell $l$. 
	The channel estimation error  $\mathbf{e}_{lk}^l = \mathbf{g}_{lk}^l - \hat{\mathbf{g}}_{lk}^l$ is independently distributed as $\mathbf{e}_{lk}^l \sim \mathcal{CN} \left( \mathbf{0}, \mathbf{C}_{lk}^l \right)$, where $\mathbf{C}_{lk}^l = \mathbf{R}_{lk}^l - \tau_p \hat{p}_{lk} \mathbf{R}_{lk}^l \pmb{\Psi}_{lk}^{-1} \mathbf{R}_{lk}^l$.

%\subsection{Downlink Data Transmission}

We use the remaining $\tau_{c} -\tau_{p}$ symbols per coherence interval for DL data transmission. 
The $n$-th data symbol that BS $l$ sends to user~$k$ in cell~$l$ is denoted $s_{lk}[n]$, where  $n$ is an index from $1$ to $\tau_{c} -\tau_{p}$. The data symbols have zero mean and normalized power: $\mathbb{E} \{ |s_{lk}[n]|^2 \} = 1$. By assing a linear precoding vector $\mathbf{w}_{lk} \in \mathbb{C}^M$ to its user~$k$, the signal that BS~$l$ sends to all users in cell~$l$ is
	\begin{equation}
	\mathbf{x}_{l}[n] = \sum_{k=1}^K \sqrt{\rho_{\rm dl}\eta_{lk}} \mathbf{w}_{lk}s_{lk}[n],
	\end{equation}
	where $\rho_{\rm dl}$ is the maximum DL transmit power and $\eta_{lk} \in [0,1]$ determines power allocation to user~$k$ in cell~$l$. We consider an arbitrary selection $\eta_{l1},\ldots,\eta_{lK}$ in every cell but note that it must be selected such that $\sum_{k=1}^{K} \eta_{lk} \leq 1$ that should hold for $l=1,\ldots,L$.

	If we define the effective DL channel gains as 
	\begin{equation} \label{eq:alphaVal}
	\alpha_{lk}^{l'k'} = \sqrt{\rho_{\rm dl}} \left(\mathbf{g}_{lk}^{l'}\right)^{\rm H} \mathbf{w}_{l'k'},
	\end{equation}
the received signal at user~$k$ in cell~$l$ can be expressed as
	\vspace{-.1in}
	\begin{equation}\label{eq:ReceivedSigv1}
	\begin{aligned}
	y_{lk} [n] &= \sqrt{\eta_{lk}}\alpha_{lk}^{lk} s_{lk}[n] + \sum\limits_{\substack{k'=1, \\k' \neq k}}^K  \sqrt{\eta_{lk'}}\alpha_{lk}^{lk'} s_{lk'}[n]
	\\
	&+ \sum\limits_{\substack{l'=1,\\ l' \neq l}}^L \sum_{k'=1}^K \sqrt{\eta_{l'k'}}\alpha_{lk}^{l'k'} s_{l'k'}[n]  + \tilde{w}_{lk}[n],
	\end{aligned}
	\end{equation}
	where $\tilde{w}_{lk}[n] \sim \mathcal{CN}(0, \sigma_{\mathrm{DL}}^2)$ is the additive noise.  The first term in \eqref{eq:ReceivedSigv1} is the desired signal for user~$k$ in cell~$l$ and the second term is the intra-cell interference. The remaining terms are inter-cell interference and noise.

	To decode the desired signal $s_{lk}[n]$, user~$k$ in cell~$l$ should know  $\alpha_{lk}^{lk}$
	and the average power of the remaining interference-plus-noise terms. Learning $\alpha_{lk}^{lk}$ is the most critical issue since its value changes in every coherence interval, thus an efficient DL channel estimation procedure is needed.
	One option is to spend a part of the coherence interval on transmitting DL pilots \cite{ngo2013massive}. 
	Another option is to utilize the structure created by the fact that the precoding vector is computed based on an MMSE estimate of $ \mathbf{g}_{lk}^{l} $. Although $\mathbb{E}\{ \mathbf{g}_{lk}^{l} \} =\mathbf{0}$, we have $\mathbb{E}\{ \alpha_{lk}^{lk} \} > 0$ for most precoding schemes, thus a basic estimate of $\alpha_{lk}^{lk}$ is its mean value $\mathbb{E}\{ \alpha_{lk}^{lk} \}$ \cite{Marzetta2006a}.
	The latter solution is attractive in ideal Massive MIMO systems where the channel hardening property implies that $\alpha_{lk}^{lk}$ is close to its mean value when the number of antennas is large \cite{jose2011pilot,Marzetta2006a}.
The drawback of  these solutions is the extra pilot overhead and the substantial performance reduction in the high-SNR regime, respectively. 	
	
	%%%%%%%%%%%%%%%%%%%%%%%%%%
	%%%%  SECTION III %%%%%%%%
	%%%%%%%%%%%%%%%%%%%%%%%%%%
	\section{Model-based Estimation Approach}\label{sec:modelbased}
	
	We want to estimate the realization of $\alpha_{lk}^{lk}$, for each user $k$ in a given cell $l$ in a blind manner, without transmitting explicit DL pilots.

	To this end, the user computes the sample mean of the received signal power in the current coherence interval:
	\vspace{-0.08in}
	\begin{equation} \label{eq:Xilko}
	\xi_{lk}  = \frac{\sum_{n=1}^{\tau_c - \tau_p} \left|y_{lk} [n]\right|^2 }{\tau_c - \tau_p}.
	\end{equation}
	The data signals and noise take new independent realizations for every $n$, thus we obtain the following result when the coherence interval is large.
	
	\begin{lemma}\label{Lemma:asymptotic}

		As $\tau_c \rightarrow \infty$ (for a fixed $ \tau_p $), $\xi_{lk}$ in \eqref{eq:Xilko} converges in probability as follows:
		\begin{equation} \label{eq:Xilk}
		\begin{split}
		\xi_{lk}\!\! \xrightarrow{P} \!\!\left( \!\!\eta_{lk}\left|\alpha_{lk}^{lk}\right|^2 \!\!+ \!\!\!\!\sum\limits_{\substack{k'=1,\\ k' \neq k}}^K  \!\!\eta_{lk'}\left|\alpha_{lk}^{lk'}\right|^2\!\! +\!\! \sum\limits_{\substack{l'=1,\\ l' \neq l}}^L \!\!\sum_{k'=1}^K \!\!\eta_{l'k'}\left|\alpha_{lk}^{l'k'}\right|^2\!\! + \!\sigma_{\mathrm{DL}}^2\!\! \right).
		\end{split}
		\end{equation}
	\end{lemma}
	\begin{proof}
		The detailed proof is provided in \cite{amin2021channel}.
	\end{proof}
	The first term at the right-hand side of \eqref{eq:Xilk} is the desired channel gain of user~$k$ in cell~$l$, the other terms are interference plus the noise variance. Note that the right-hand side of \eqref{eq:Xilk} is constant within a coherence interval but takes different independent realizations in different blocks. Hence, the convergence in probability in \eqref{eq:Xilk} refers to the randomness of the signals and noise, but is conditioned on the channel realizations in the considered coherence interval.
	Our goal is to utilize the asymptotic limit in \eqref{eq:Xilk} to estimate $\alpha_{lk}^{lk}$  from $\xi_{lk}$, but this is an ill-posed estimation problem since there are $LK$ unknowns: $\alpha_{lk}^{l'k'}$, $l'=1,\ldots,L$, $k'=1,\ldots,K$.
	To resolve this issue, we will make use of another asymptotic result, based on the regime where the number of users per cell is large.

	\begin{lemma} \label{lemma:asymptoicLLN}
	Suppose the users are dropped in each cell independently at random according to some common distribution for which $\alpha_{lk}^{l'k'}$ has bounded variance.
	As $\tau_c,K \to \infty$ such that $K/\tau_c \to 0$ and $\tau_p = fK$, we obtain the following asymptotic equivalence:  	
	\vspace{-0.08in}	
		\begin{equation}\label{eq:aymptotLemma}
		\begin{split}
		\frac{1}{K }\xi_{lk}  \asymp \frac{1}{K} \left(\eta_{lk}\left|\alpha_{lk}^{lk}\right|^2 + \sum\limits_{\substack{k'=1,\\ k' \neq k}}^K  \eta_{lk'} \mathbb{E} \left\{ \left|\alpha_{lk}^{lk'}\right|^2 \right\}
		\right. \\ \left.+ \sum\limits_{\substack{l'=1,\\ l' \neq l}}^L \sum\limits_{k'=1}^K \eta_{lk'}  \mathbb{E} \left\{ \left|\alpha_{lk}^{l'k'}\right|^2 \right\}
		+ \sigma_{\mathrm{DL}}^2\right).
		\end{split}
		\end{equation}
	\end{lemma}
	\begin{proof}  
	 The detailed proof is provided in \cite{amin2021channel}.
	\end{proof}	
	Lemma \ref{lemma:asymptoicLLN} implies that the mutual interference terms can be replaced by their mean values as $K \rightarrow \infty$ and  the mean value is computed with respect to the channel realizations for given user locations. 
This is a rigorous asymptotic result but we will utilize it as a motivation for approximating $\xi_{lk}$ for a finite number of users $K$ per cell as follows:
	\begin{equation} \label{eq:Xilkv1Orig}
	\begin{split}
	\xi_{lk}  \approx \left(\eta_{lk}\left|\alpha_{lk}^{lk}\right|^2 + \sum\limits_{\substack{k'=1,\\ k' \neq k}}^K  \eta_{lk'} \mathbb{E} \left\{ \left|\alpha_{lk}^{lk'}\right|^2 \right\}
	\right.\\\left.+  \sum\limits_{\substack{l'=1,\\ l' \neq l}}^L \sum\limits_{k'=1}^K  \eta_{l'k'} \mathbb{E} \left\{ \left|\alpha_{lk}^{l'k'}\right|^2 \right\}
	+ \sigma_{\mathrm{DL}}^2\right).
	\end{split}
	\end{equation}
	If there would be equality in \eqref{eq:Xilkv1Orig}, we can solve for $|\alpha_{lk}^{lk}|$:
	\begin{equation}
	\alpha_{lk}^{lk} \approx |\alpha_{lk}^{lk}| \approx \sqrt{\frac{ \xi_{lk}  - T_{lk}  }{\eta_{lk}}}
	\end{equation}
	where we also utilize that $\alpha_{lk}^{lk}$ is approximately positive and 
	\begin{equation}\label{eq:TLK}
	\begin{aligned}
	T_{lk} &= \!\!\sum_{\substack{k'=1,\\ k' \neq k}}^K\!\!\!  \eta_{lk'}  \mathbb{E} \left\{ \left|\alpha_{lk}^{lk'}\right|^2 \right\} + \sum_{\substack{l'=1,\\ l' \neq l}}^L \!\sum_{k'=1}^K \!\!\!\eta_{l'k'} \mathbb{E} \left\{ \left|\alpha_{lk}^{l'k'}\right|^2 \right\} + \sigma_{\mathrm{DL}}^2.
	\end{aligned}
	\end{equation}
	Based on \eqref{eq:Xilkv1Orig}, we propose the following estimator
		\vspace{-.07in}
	\begin{equation} \label{eq:alphalkhat}
	\hat{\alpha}_{lk}^{lk} = \begin{cases}
	\sqrt{\frac{ \xi_{lk}  - T_{lk}  }{\eta_{lk}}}, & \mbox{if } \xi_{lk} > \Theta_{lk} ,  \\
	\mathbb{E} \big\{ \alpha_{lk}^{lk} \big\},& \mbox{otherwise}.
	\end{cases}
	\end{equation}
	The second case utilizes the mean value as the estimate of ${\alpha}_{lk}^{lk}$ when $\xi_{lk}$ is below some threshold $ \Theta_{lk} \geq T_{lk}$ that identifies the cases when the proposed estimator is inaccurate.

	We can measure the accuracy of this estimator using the normalized MSE, defined at the user $k$ in cell $l$ as
	\begin{equation} \label{eq:NMSE}
	\text{MSE}_{lk} = \frac{\mathbb{E}\{|\hat{\alpha}_{lk}^{lk}-{\alpha}_{lk}^{lk}|^2\}}{\mathbb{E}\{|{\alpha}_{lk}^{lk}|^2\}}.
	\end{equation}

	We can compute $T_{lk}$ as follows when MR precoding is used.

	\begin{lemma} \label{Lemma:ClosedFormMRCorrelated}
	If MR precoding with $	\mathbf{w}_{lk} = \frac{\hat{\mathbf{g}}^{l}_{lk}}{\sqrt{\mathbb{E}\left\{\|\hat{\mathbf{g}}^{l}_{lk} \|^{2}\right\}}}$
	is utilized, then we can estimate ${\alpha}_{lk}^{lk}$ using \eqref{eq:alphalkhat} with
		\begin{equation} \label{eq:Tlk-derivation}
		\begin{aligned}
		&T_{lk} =   \sum\limits_{\substack{k'=1, \\k' \neq k}}^K \rho_{\rm dl}\eta_{lk'} \frac{\mathrm{Tr}\left(\mathbf{R}^{l}_{lk'}\boldsymbol{\Psi}^{-1}_{lk'}\mathbf{R}^{l}_{lk'}\mathbf{R}^{l}_{lk}\right)}{\mathrm{Tr}\left(\mathbf{R}^{l}_{lk'}\boldsymbol{\Psi}^{-1}_{lk'}\mathbf{R}^{l}_{lk'}\right)}\\
		&+ \sum\limits_{\substack{l'=1,\\ l' \neq l}}^L \sum_{k'=1}^K \rho_{\rm dl}\eta_{l'k'} \frac{\mathrm{Tr}\left(\mathbf{R}^{l'}_{l'k'}\boldsymbol{\Psi}^{-1}_{l'k'}\mathbf{R}^{l'}_{l'k'}\mathbf{R}^{l'}_{lk}\right)}{\mathrm{Tr}\left(\mathbf{R}^{l'}_{l'k'}\boldsymbol{\Psi}^{-1}_{l'k'}\mathbf{R}^{l'}_{l'k'}\right)} \\
		&+ \sum\limits_{l' \in \mathcal{P}_{l} \setminus \{l\}} \rho_{\rm dl}\eta_{l'k} \left(\frac{\hat{p}_{lk}\tau_{p} \left| \mathrm{Tr}\left(\mathbf{R}^{l'}_{lk}\boldsymbol{\Psi}^{-1}_{lk}\mathbf{R}^{l'}_{l'k} \right)\right|^{2}  }{ \mathrm{Tr}\left(\mathbf{R}^{l'}_{l'k}\boldsymbol{\Psi}^{-1}_{lk}\mathbf{R}^{l'}_{l'k}\right)} \right)  + \sigma_{\mathrm{DL}}^2.
		\end{aligned}
		\end{equation}

	\end{lemma}
	\begin{proof}
		The detailed proof is provided in \cite{amin2021channel}
	\end{proof}

	\subsection{Ergodic SE }\label{sec:SE}

To evaluate the SE achieved when using the proposed estimator in \eqref{eq:alphalkhat}, we need to derive a new SE expression because the DL effective channel gain estimate is correlated with the data symbols (which is not supported by the conventional SE expressions).
 
 To resolve this issue, for the $n$-th data symbol, we remove $y_{lk}[n]$ from the received data and the sample average power of the signal at user~$k$ in cell~$l$ is reformulated as \cite{ngo2017no}
	\begin{equation} \label{eq:Estv2}
	\xi^{'}_{lk}[n]  = \frac{\sum_{n'=1,n' \neq n}^{\tau_c - \tau_p} \left|y_{lk} [n']\right|^2 }{\tau_c - \tau_p-1}.
	\end{equation}
	Utilizing \eqref{eq:Estv2} to estimate ${\alpha}_{lk}^{lk}$, denoted as $\bar{\alpha}_{lk}^{lk}[n]$, it is clear that  $\bar{\alpha}_{lk}^{lk}[n]$ is close to $\hat{\alpha}_{lk}^{lk}[n]$ when $\tau_c - \tau_p$ grows large. By dividing \eqref{eq:ReceivedSigv1} with $\sqrt{\eta_{lk}}\bar{\alpha}_{lk}^{lk}[n]$ to perform equalization of the effective channel gains (i.e., making the factor in front of $s_{lk}[n]$ approximately equal to one), we obtain the received signal as	

	\begin{equation}
	\begin{aligned} 
	&y'_{lk}[n] = \!\!\mathbb{E}\left\{\frac{\alpha_{lk}^{lk}}{\bar{\alpha}_{lk}^{lk}[n]}\right\} s_{lk}[n] + \left(\frac{\alpha_{lk}^{lk}}{\bar{\alpha}_{lk}^{lk}[n]} - \mathbb{E}\left\{\frac{\alpha_{lk}^{lk}}{\bar{\alpha}_{lk}^{lk}[n]}\right\}\right) s_{lk}[n] \\
	&+ \sum\limits_{\substack{k'=1,\\ k' \neq k}}^K  \sqrt{\frac{\eta_{lk'}}{\eta_{lk}}}\frac{\alpha_{lk}^{lk'}}{\bar{\alpha}_{lk}^{lk}[n]} s_{lk'}[n] 
	+ \sum\limits_{\substack{l'=1,\\ l' \neq l}}^L \sum_{k'=1}^K \sqrt{\frac{\eta_{l'k'}}{\eta_{lk}}}\frac{\alpha_{lk}^{l'k'}}{\bar{\alpha}_{lk}^{lk}[n]} s_{l'k'}[n] \\
	&+ \frac{\tilde{w}_{lk}[n]}{\sqrt{\eta_{lk}}\bar{\alpha}_{lk}^{lk}[n]}
\label{eq:normalizedReceived}
	\end{aligned}
		\end{equation}
	where the first term is the desired signal $s_{lk}[n]$ multiplied with a deterministic channel gain. For a successful equalization, the second term is small. By treating the last four terms as additive noise and applying the channel capacity bounding technique developed in \cite{jose2011pilot}, we obtain the following result.
	
\begin{lemma}\label{lemma:SE}
	A DL ergodic SE for user $k$ in cell $l$ is	
			\begin{equation} \label{eq:SpectralEfficiencyH}
			{\rm{R}}_{lk} = \left( 1 -\frac{\tau_p}{\tau_c} \right) \log_2 \left( 1 + \mathrm{SINR}_{lk} \right), \mbox{ [b/s/Hz]},
			\end{equation}
			%\hrulefill
			%\end{figure*}
			where the effective DL signal to interference and noise ratio (SINR) is given in \eqref{eq:SINR} on the top of the next page.
		\begin{figure*}
		\begin{equation}\label{eq:SINR}
		\mathrm{SINR}_{lk}= \frac{\left|\mathbb{E}\left\{\frac{\alpha_{lk}^{lk}}{\bar{\alpha}_{lk}^{lk}[n]}\right\}\right|^2}{\mathrm{var}\left\{\frac{\alpha_{lk}^{lk}}{\bar{\alpha}_{lk}^{lk}[n]}\right\}+ \sum\limits_{\substack{k'=1,\\ k' \neq k}}^K \frac{\eta_{lk'}}{\eta_{lk}} \mathbb{E}\left\{\left|\frac{\alpha_{lk}^{lk'}}{\bar{\alpha}_{lk}^{lk}[n]}\right|^2\right\} + \sum\limits_{\substack{l'=1,\\ l' \neq l}}^L \sum\limits_{k'=1}^K \frac{\eta_{l'k'}}{\eta_{lk}}\mathbb{E}\left\{\left|\frac{\alpha_{lk}^{l'k'}}{\bar{\alpha}_{lk}^{lk}[n]}\right|^2\right\}+ \frac{\sigma_{\mathrm{DL}}^2}{\eta_{lk}}\mathbb{E} \left\{\frac{1}{\left|\bar{\alpha}_{lk}^{lk}[n]\right|^2}\right\} }.
		\end{equation} \hrule
%				\vspace{-0.35in}
		\end{figure*}	
	\end{lemma}	

	This is an achievable SE, in other words, a lower bound on the ergodic channel capacity.

	For benchmark purposes, we will also consider the ideal case when the users have access to perfect CSI. Then, the first term in \eqref{eq:ReceivedSigv1} is the desired signal multiplied with a known channel and the remaining terms can be treated as additive noise. By applying a standard ergodic channel capacity bounding technique from \cite{redbook}, we have the following result.

	\begin{lemma} \label{lemma:SEPerfect}
		If perfect CSI is available at the user, then the DL ergodic spectral efficiency given as
		%provided in \eqref{eq:SpectralEfficiencyH}.
		%\begin{figure*}
		\begin{equation} \label{eq:SpectralEfficiencyPerfect}
		{\rm{R}}_{lk} = \left( 1 -\frac{\tau_p}{\tau_c} \right) \mathbb{E}\left\{\log_2 \left( 1 + \mathrm{SINR}_{lk} \right)\right\}, \mbox{ [b/s/Hz]},
		\end{equation}
		%\hrulefill
		%\end{figure*}
		where the SINR is given as
		\begin{equation}\label{eq:SinrPerfect}
		\mathrm{SINR}_{lk}= \frac{\eta_{lk}\left|{\alpha_{lk}^{lk}}\right|^2}{ \sum\limits_{\substack{k'=1,\\ k' \neq k}}^K \eta_{lk'} \left|\alpha_{lk}^{lk'}\right|^2 + \sum\limits_{\substack{l'=1,\\ l' \neq l}}^L \sum\limits_{k'=1}^K \eta_{l'k'}\left|\alpha_{lk}^{l'k'}\right|^2+ \sigma_{\mathrm{DL}}^2 }.
		\end{equation}
	\end{lemma}

\section{Deep-learning-based Estimation Approach }\label{sec:datadriven}
%for Downlink Massive MIMO Channel Estimation

 The proposed blind DL channel estimator in \eqref{eq:alphalkhat} is model-aided, in the sense that it was developed by studying the asymptotic properties of the system model. While the estimator is expected to work well when the coherence interval is large and there are many users per cell, there is no guarantee that the estimator will work well under the circumstances that occur in practical Massive MIMO systems. For example, the number of users per cell might be small, in particular, under low-traffic hours or when the coherence interval is relatively small.  
To obtain a more practical solution, 
 we propose a deep-learning-based approach to DL channel estimation in Massive MIMO systems, where deep learning is used to "learn" an estimator in the sense of identifying a mapping between the available information at the UE and the DL effective channel gain. We tackle the mentioned limitations of the proposed model-aided blind DL channel estimator by training a fully-connected neural network for the same task. The goal is to determine under what conditions and to what extent the proposed model-aided estimator can be outperformed.
 
 The universal approximation theorem states that one can approximate any continuous function between a given input vector and the desired output vector arbitrarily well using a sufficiently large fully-connected neural network \cite{goodfellow2016deep}. However, this theorem does not provide any exact details on the neural network structure (e.g., the number of layers and neurons) or what algorithms to utilize to find the optimal approximation. This effort must be carried out for every problem at hand. Here, we utilize a fully-connected feed-forward NN to estimate $\alpha^{lk}_{lk}$ from input data available at an arbitrary user $k$ in cell $l$. As input to the NN, we consider three features: $\xi^{'}_{lk}[n]$ that is given in \eqref{eq:Estv2}, $T_{lk}$ provided in \eqref{eq:TLK}, and $\eta_{lk}\rho_{\rm dl}\beta^{l}_{lk}$ for the user $k$ in cell $l$.
 {\color{black}{Note that we can use \eqref{eq:TLK} with knowing the covariance matrices or the pilot reuse pattern in \eqref{eq:Tlk-derivation}, which are non-trivial to determine at the users' side.}} 
  The input is selected to enable the network to learn about the pathloss model, propagation environment, and mapping between the sample average power and effective channel gains. The input vector to the neural network is denoted as $\boldsymbol{\kappa}_{lk} \in \mathbb{R}^{P}$, where $P=3$ in the proposed design. The output is a scalar  $o_{lk}$ that is supposed to be equal to the absolute value  of DL effective channel gain $\alpha^{lk}_{lk}$. The NN has $I = 3$ hidden layers with given size specified in Table \ref{table:1}, to approximate the ideal non-linear mapping from $\boldsymbol{\kappa}_{lk}$ to $o_{lk}$ \cite{o2017introduction}.% the size of each layer is specified in Table \ref{table:1}.
 
 To design the NN, we fine-tuned some of the network's parameters, such as the number of layers, the number of nodes per hidden layer, learning rate, activation functions, etc., experimentally to find a network structure that offers good performance in terms of NMSE. The rectified linear unit (ReLU) was  selected as the successful candidate activation function of hidden layers. The detailed information about the layout is provided in Table \ref{table:1} and the other parameters settings for the deep learning algorithm are provided in Section \ref{sec:numerical-results}.

The network is trained using a set of labeled training data consisting of inputs and corresponding optimal outputs pairs i.e., defined as $\{\boldsymbol{\kappa}^{d}_{lk}, \hat{o}^{d}_{lk}\}^{D}_{d=1}$, where $D$ is the number of points in the set. For each $d$, $ \boldsymbol{\kappa}^{d}_{lk}$ is the input vector and the corresponding desired output is $ \hat{o}^{d}_{lk}$ \cite{o2017introduction}. We train the network for a \emph{typical} user so that the same trained network is applicable for all users. The training is done offline, but the NN is used by the users in a cellular system, in online mode. In addition, the trained NN should be generalizable meaning that one can use the same model to approximate the correct output not only on the training data but also on any other input data vector generated from the same distribution as the training input data. The data is generated from the simulation setup, but it is possible to obtain such data from measurements in a practical setup. The main challenge is to obtain the labels, but one feasible solution is to occasionally transmit orthogonal pilot sequences of length $\tau_{c}-\tau_{p}$, in an entire coherence interval in the DL. These pilot sequences can be reused sparsely in the network (e.g., reuse 7) so that there is essentially no pilot contamination, and the SNR will be very high after despreading, so that the true $\alpha^{lk}_{lk}$ can be estimated accurately. These sequences can also be utilized to estimate and calibrate other aspects of the system.

\begin{table}[t!]
%	\vspace{-.2in}
	\centering
	\caption{Layout of the NN.} 
	\begin{center}
		\begin{tabular}{| c| c| c| c|}
			\hline
			& Neurons & Parameters& Activation function\\ 
			%\hline
			%Input & - &  - & - \\
			\hline
			Layer 1 & 32 &  256 & ReLU \\
			\hline
			Layer 2 & 64 & 2112 & ReLU \\
			\hline
			Layer 3 & 64 & 4160 & ReLU \\
			%	\hline
			%Layer 4 & 20 & 820 & Relu \\
			\hline
			%Layer 4 & 18 & 594 & Linear \\
			%\hline
		\end{tabular}\label{table:1}
	\end{center} 
\vspace{-.2in}
\end{table}

To evaluate the SE achieved when using the deep-learning-based  approach, we can utilize a similar DL ergodic SE expression as given in Lemma \ref{lemma:SE}.

\section{Numerical Results} \label{sec:numerical-results}

	We evaluate the proposed estimators by considering a multi-cell Massive MIMO setup with $4$ square cells in a grid layout in a 500\,m $\times$ 500\,m area. We use the wrap-around technique to avoid edge effects. Each BS has $M = 64$ antennas and serves $K$ users, which are uniformly distributed in their coverage area with a minimum distance of $35\,$m and  $\tau_{c}=500$ symbols. The large-scale fading coefficients are modeled as \cite{bjornson2017massive} %$\beta^{l}_{l'k'} \left[{\rm dB}\right] = -35 - 36.7\log_{10}\left(d^{l}_{l'k}/1\,\rm{m}\right) + F^{l}_{l'k'} ,$
	%============================eq 
	\begin{equation}\label{eq:largeScale}
	\beta^{l}_{l'k'} \left[{\rm dB}\right] = -35 - 36.7\log_{10}\left(d^{l}_{l'k}/1\,\rm{m}\right) + F^{l}_{l'k'} ,
	\end{equation}
	%=============================
	where $d^{l}_{l'k'}$ is the distance from user $k'$ in cell $l'$ to BS $l$ and $F^{l}_{l'k'}$ is log-normal shadow fading with a standard deviation of $7\,$dB. The noise variance is $-94\,$dBm. We assume an equal power allocation scheme in the DL data transmission, and the uplink transmit power of the users is set to $100\,$mW. Each BS equipped with a horizontal uniform linear array with half-wavelength antenna spacing and the spatial correlation matrix of user $k$ located in cell $l'$ to the BS $l$ is modeled by the approximate Gaussian local scattering model provided in \cite[Ch.~2.6]{bjornson2017massive} with the $(m,n)$th elements given by
	\begin{equation} \label{eq:LSM}
	\left[\mathbf{R}^{l}_{l'k}\right]_{m,n}  =\beta^{l}_{l'k} e^{\pi j (m-n)\sin (\varphi^{l}_{l'k})} e^{- \frac{\sigma^2_{\varphi}}{2} (\pi (m-n)\cos (\varphi^{l}_{l'k}))^2}.
	\end{equation}
	%=============================
 
	In this expression, $\varphi^{l}_{l'k}$ is the nominal angle of arrival (AoA) and the multipath components are Gaussian distributed around nominal AoA with an angular standard deviation (ASD) $\sigma_{\varphi} = 7$ degree.  For the deep-learning-based approach, the entire data set consists of $D = 1000000$ input-output vector pairs for a typical user $k$ randomly located in cell $l$ for $1000$ realizations of large-scale fading and $1000$ small-scale fading. We selected $400000$ for training, $100000$ for validation, and the rest of $500000$ for the testing phase. The implementation was carried out using the Keras library in Python. In the training phase, we selected the Adam optimizer \cite{kingma2017adam} and the loss function was the mean absolute error (MAE), the learning rate was $0.01$, the batch size was $128$, and the number of epochs was $200$.
We evaluate the performance of the proposed estimators in terms of the NMSE in the training phase as well as the SE in the data transmission phase, to investigate whether an improved NMSE also results in an improved SE.

In Fig.~\ref{fig:correlatedMRMSE}, we plot the CDF of the NMSE when the median DL SNR, i.e., $\rm{SNR_{dl}}$ of a cell-edge user is $10\,$dB. We compare the two proposed approaches against two different benchmarks: the "Hardening bound" uses $\mathbb{E}\{ \alpha_{lk}^{lk} \}$ as the estimate of $\alpha_{lk}^{lk}$, \cite{sanguinetti2019toward, yang2017massive} and "$\tau_c = \infty$" assumes that the user knows the asymptotic value of $\xi_{lk}$. The hardening bound result is the rightmost which shows that both proposed approaches perform substantially better. The deep-learning-based approach provides the smallest NMSEs, particularly for the most unfortunate users. {\color{black}Due to pilot contamination and the i.i.d. fading assumption, there will always exist channel estimation errors, even in the limiting regime. The performance of ``model-aided'' and ``$\tau_c = \infty$" coincide, which shows that the model-aided solution can achieve good performance even with finite radio resources.}

Figs.~\ref{fig:SE_MR_K310} and \ref{fig:SE_MR_K1010} show the CDF of the SE per user for $K = 3$ and $K = 10$, respectively.
Fig.~\ref{fig:SE_MR_K310} shows a significant SE improvement for the model-aided approach compared to the hardening bound, which implies that the conventional hardening bound greatly underestimates the achievable SE when the channel hardening is limited, as is the case in the considered channel model with a small ASD. The deep-learning-based approach results in higher SE than the model-aided approach in the lower 40 \% of the CDF curve and comparable SE for the other 60 \%. We also show the SE obtained with perfect CSI at the user, based on Lemma \ref{lemma:SEPerfect}, and there is a significant difference. In Fig.~\ref{fig:SE_MR_K1010} the gap between the proposed approaches and perfect CSI is reduced. The estimated effective channel gain is getting closer to its asymptotic limit by increasing the number of users, resulting in a comparable performance for perfect CSI and the hardening bound. By comparing with $K = 3$, the SEs are decreasing, which shows that interference is becoming more dominant which is also affecting the result of  perfect CSI. The results of the deep-learning-based approach for $K=10$ are obtained by using the trained model for $K=3$, which indicates that the  deep-learning-based approach is robust towards changes in the number of users.

\begin{figure}[!hbt]
		\centering
		\includegraphics[width=0.85\columnwidth]{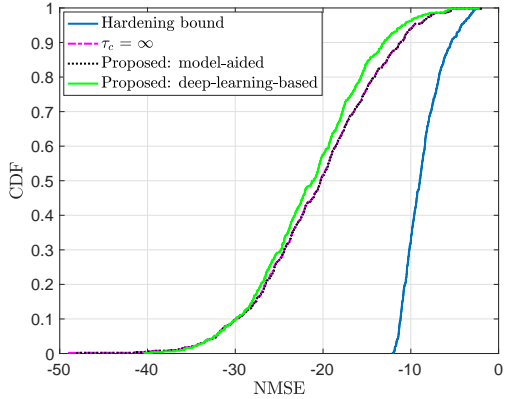}
		\caption{Comparison of NMSE for different estimation approaches.}
		\label{fig:correlatedMRMSE}
			\vspace{-.1in}
\end{figure}

	\begin{figure}[!hbt]
		\centering
		\includegraphics[width=0.85\columnwidth]{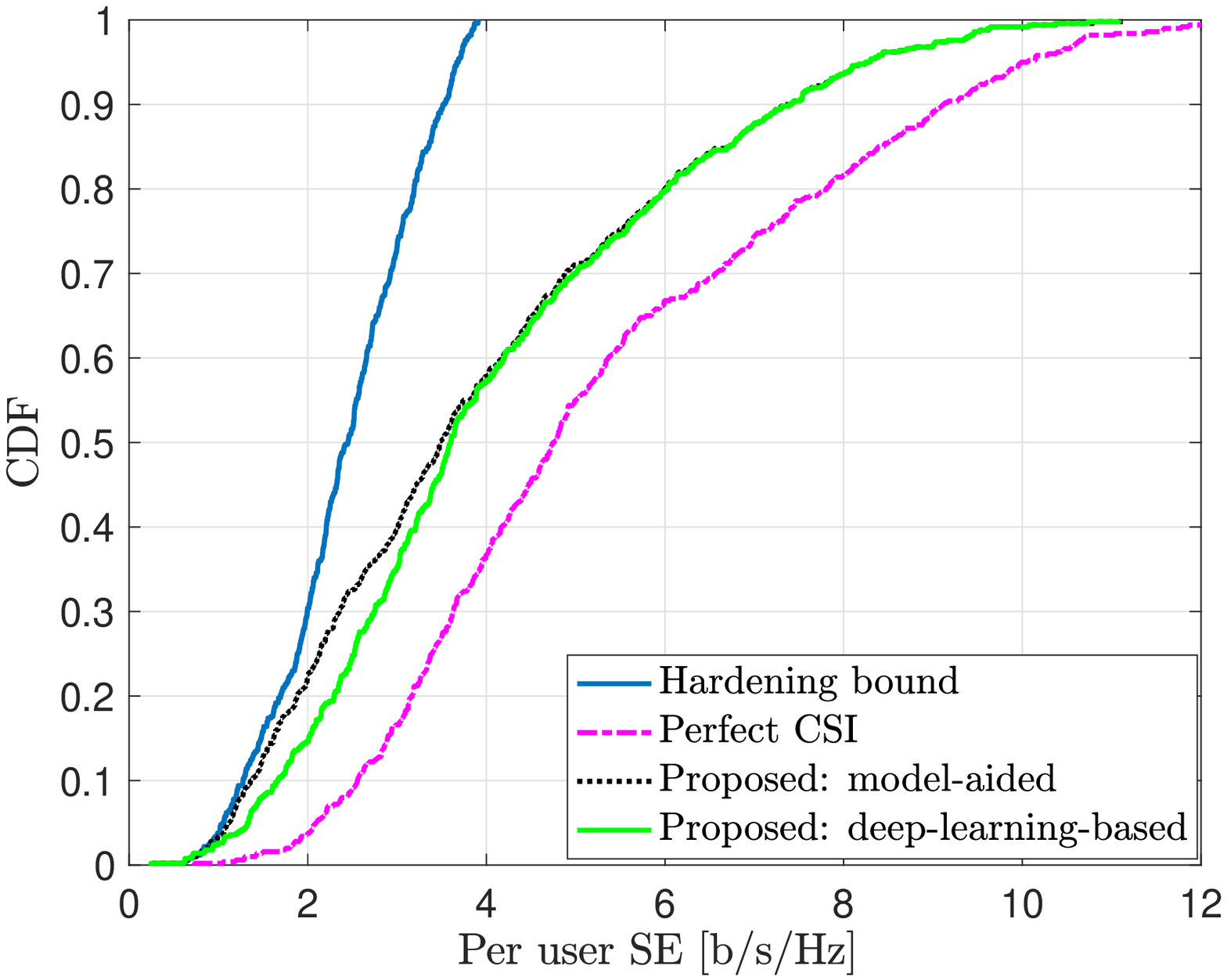}
		\caption{CDF of the SE per user with $K = 3$.}
		\label{fig:SE_MR_K310}
	\end{figure}
\begin{figure}[!hbt]
	\centering
	\includegraphics[width=0.85\columnwidth]{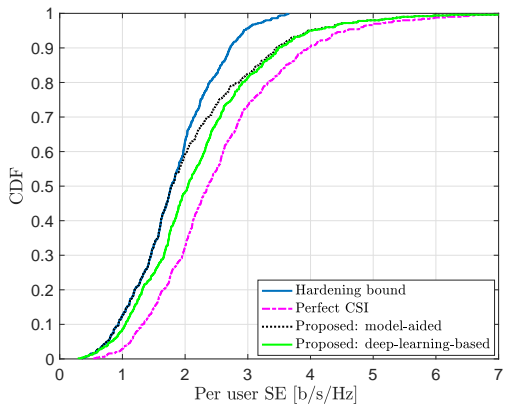}
	\caption{CDF of the SE per user with $K = 10$.}
	\label{fig:SE_MR_K1010}
\end{figure}

\section{Conclusion}\label{sec:conclusion}
This paper proposed a new model-aided approach and a new deep-learning-based approach to the estimation of the DL effective channel gains
 in multi-cell Massive MIMO systems. 
 The former approach is based on a closed-form expression that was obtained using asymptotic analysis, while the latter approach is based on supervised training  of a neural network.
 We compared the proposed approaches to the conventional approach of utilizing the mean value of the effective channel gains as the estimate, which only works well when there is a high level of channel hardening. 
The proposed approaches provide superior estimation quality (NMSE) and communication performance (SE) for channels with a low level of channel hardening, which happens in practical environments with limited scattering.

\vspace{-3mm}

	\bibliographystyle{IEEEtran}
	\bibliography{refs}
\end{document}